\documentclass[11pt, a4paper]{amsart}

\usepackage{fullpage}
\usepackage[foot]{amsaddr}
\usepackage[bookmarks=true,hypertexnames=false,pagebackref]{hyperref}
\hypersetup{colorlinks=true, citecolor=blue, linkcolor=red, urlcolor=blue}
\usepackage[nott,sfmathbb]{kpfonts}
\usepackage[T1]{fontenc} % best for Western European languages
\usepackage{textcomp} % required to get special symbols
\usepackage{mathalfa}
\usepackage{bm}
\usepackage{varwidth}
\usepackage[textsize=tiny]{todonotes}
\usepackage{mdframed}
\usepackage{graphicx}

\usepackage{amsthm,amsmath,amssymb,physics}
\usepackage{xifthen}
\usepackage[ruled,vlined, linesnumbered]{algorithm2e}

\usepackage{hyperref}
\usepackage{cleveref} % Load algorithm2e before cleveref in order to ref algorithms correctly
\usepackage{cite}
\newtheorem{theorem}{Theorem}

\newtheorem{lemma}[theorem]{Lemma}

\newtheorem{definition}[theorem]{Definition}

\newtheorem{proposition}[theorem]{Proposition}

\def\*#1{\mathbf{#1}} \def\+#1{\mathcal{#1}} \def\-#1{\mathrm{#1}}\def\^#1{\mathbb{#1}}\def\!#1{\texttt{#1}}
\def\C{\mathscr{C}}
\def\B{\mathscr{B}}
\newcommand{\set}[1]{\left\{#1\right\}}
\newcommand{\tuple}[1]{\left(#1\right)} 
\newcommand{\eps}{\varepsilon}
 \newcommand{\tp}{\tuple}
\renewcommand{\mid}{\;\middle\vert\;} \newcommand{\cmid}{\,:\,}

\renewcommand{\Pr}[1]{ \mathbf{Pr}\left[#1\right]}
\newcommand{\E}[1]{ \mathbf{E}\left[#1\right]}
\renewcommand{\emptyset}{\varnothing}

\newcommand{\LLL}{Lov\'asz local lemma}
\title{A Perfect Sampler for Hypergraph Independent Sets}

\author{Guoliang Qiu}
\address[Guoliang Qiu]{Shanghai Jiao Tong University, China. \textnormal{E-mail: \url{guoliang.qiu@sjtu.edu.cn}}}

\author{Yanheng Wang}
\address[Yanheng Wang]{ETH Z\"{u}rich, Switzerland. \textnormal{E-mail: \url{yanhwang@student.ethz.ch}}}
\author{Chihao Zhang}
\address[Chihao Zhang]{Shanghai Jiao Tong University, China. \textnormal{E-mail: \url{chihao@sjtu.edu.cn}}}

\thanks{Part of this work was done when Yanheng Wang was an undergraduate student at Shanghai Jiao Tong University.}

\begin{document}

\begin{abstract}
    The problem of uniformly sampling hypergraph independent sets is revisited. We design an efficient \emph{perfect} sampler for the problem under a condition similar to that of the asymmetric \LLL. When applied to $d$-regular $k$-uniform hypergraphs on $n$ vertices, our sampler terminates in expected $O(n\log n)$ time provided $d\le c\cdot 2^{k/2}/k$ for some constant $c>0$. If in addition the hypergraph is linear, the condition can be weaken to $d\le c\cdot 2^{k}/k^2$ for some constant $c>0$, matching the rapid mixing condition for Glauber dynamics in Hermon, Sly and Zhang \cite{HSZ19}.
\end{abstract}

\maketitle

\section{Introduction}\label{sec1}

The problem of uniformly sampling hypergraph independent sets, or equivalently the solutions of monotone CNF formulas, has been well-studied in recent years. Consider a hypergraph $\Phi=(V,\C)$ on $n := \abs{V}$ vertices. A set $S\subseteq V$ is an independent set if $C\cap S\ne C$ for all $C\in \C$. Assuming the hypergraph is $d$-regular and $k$-uniform, \cite{HSZ19} showed that the natural Glauber dynamics mixes in $O(n\log n)$ time when $d\le c\cdot 2^{\frac{k}{2}}$ for some constant $c>0$. The sampler implies a fully polynomial-time randomized approximation scheme (FPRAS) for counting hypergraph independent sets \cite{JVV86}. On the other hand, it was shown in \cite{BGGGS19} that there is no FPRAS for the problem when $d\ge 5\cdot 2^{\frac{k}{2}}$, unless $\mathbf{NP}=\mathbf{RP}$. Therefore, the result of \cite{HSZ19} is tight up to a multiplicative constant.

The proof in \cite{HSZ19} analyzes the continuous-time Glauber dynamics under the framework of \emph{information percolation} developed in \cite{LS16} for studying the cutoff phenomenon of the Ising model. In this framework, one can view the coupling history of a Markov chain as time-space slabs, and the failure of the coupling at time $t$ as a discrepancy path percolating from $t$ back to the beginning. The analysis of this structure can result in (almost) optimal bounds in many interesting settings.

The percolation analysis in \cite{HSZ19} is technically complicated due to the \emph{continuous} nature of the chain which leads to involved dependencies in both time and space. Recently, discrete analogs of the time-space slabs have been introduced in sampling solutions of general CNF formulas in the \LLL~regime \cite{JPV21,HSW21}. Notably in \cite{HSW21}, an elegant discrete time-space structure, tailored to \emph{systematic scan}, was introduced to support the analysis of \emph{coupling from the past} (CFTP) paradigm \cite{PW96}. In contrast to simulating Glauber dynamics, the method of CFTP allows producing \emph{perfect samples} (i.e. without approximation error) from the stationary distribution.

In this article, we refine the discrete time-space structure in \cite{HSW21}, which we call the \emph{witness graph}, and apply it to the problem of sampling hypergraph independent sets. This leads to an efficient sampler that (1) outputs perfect samples; (2) applies to instances satisfying an asymmetric \LLL-like condition and matches the bound in \cite{HSZ19} for $d$-regular $k$-uniform linear hypergraphs; and (3) is very simple to analyze. More specifically, we study a natural grand coupling of the systematic scan for sampling hypergraph independent sets. To apply the method of CFTP, one needs to detect the coalescence of the grand coupling efficiently at each stage of the algorithm. The monotone property of hypergraph independent sets allows us to reduce the detection to arguing about percolation on the witness graph. This observation provides sufficient flexibility for us to carefully bound related probabilities.  We first show that, under a condition similar to the \emph{asymmetric} \LLL, a perfect sampler exists.

\begin{theorem}\label{thm:asym-main}
  Let $\+G$ collect all hypergraphs $\Phi=(V,\C)$ such that
    \[
		\forall C\in \C: 2\abs{C}\cdot 2^{-\abs{C}}\le (1-\eps) \cdot x(C)\cdot \prod_{C'\in\Gamma_{\Phi^2}(C)}(1-x(C')),
	\]
	for some constant $\eps\in(0,1)$ and function $x:\C\to(0,1)$. Here $\Gamma_{\Phi^2}(C)$ denotes the set of hyperedges within distance $2$ to $C$ in $\Phi$ (excluding $C$ itself).
	
	There exists an algorithm such that given as input a hypergraph $\Phi\in\+G$, outputs an independent set of $\Phi$ uniformly at random, in expected time \[ O \left( -\frac{1}{\log(1-\eps)}\cdot \log (\sum_{C\in\C} \frac{x(C)}{1-x(C)})\cdot \sum_{C \in \C}  \sum_{C \cap C' \neq \emptyset} \abs{C'}  \abs{C} \right) , \]
	where $d_v$ is the number of hyperedges in $\Phi$ containning $v$.
\end{theorem}

In the proof of \Cref{thm:asym-main}, we view a discrepancy path of the coupling percolating from time $t$ back to the beginning as an object similar to the \emph{witness tree} in \cite{MT10} for certifying the non-termination of a randomized algorithm. This is an interesting analogy between an algorithm that \emph{samples} the solutions of CNF formulas and an algorithm that \emph{finds} them. We map the discrepancy paths in the witness graph to connected trees generated by \emph{multi-type Galton-Watson process} in the hypergraph. As a result, a certain spatial mixing property implies the rapid mixing of the chain.

For $k$-uniform $d$-regular hypergraphs, the result in \Cref{thm:asym-main} translates to the condition $d\le \frac{c\cdot 2^{k/2}}{k^{1.5}}$ for some constant $c>0$ by choosing $x(C) = \frac{1}{d^2k^2}$. We give a refined analysis for this symmetric case:

\begin{theorem}\label{thm:sym-main}
	 Let $\+G$ collect all $k$-uniform $d$-regular hypergraphs such that
    \[
		d \leq  \frac{(1-\eps)}{20}\cdot \frac{2^{k/2}}{k}
	\]
	for some constant $\eps\in(0,1)$. There exists an algorithm such that given as input a hypergraph $\Phi=(V,\C)\in\+G$, outputs a hypergraph independent set of $\Phi$ uniformly at random, in expected time $O \left(-\frac{1}{\log(1-\eps)}\cdot k^2d^2n \cdot( \log n+\log d) \right)$.
\end{theorem}

Our analysis inductively enumerates discrepancy paths, taking into account the structure of the overlapping between hyperedges. %Thanks to the clean structure of the witness graph, our proof is simple and seems possible to match the one in \cite{HSZ19}. However, there are some limitations to our current analysis towards this goal.

A hypergraph is \emph{linear} if no two hyperedges share more than one vertex. We show that our perfect sampler is efficient on $k$-uniform $d$-regular linear hypergraphs. The bound is within a $O\tp{\frac{1}{k}}$ factor to the uniqueness threshold~\cite{BGGGS19}, mathcing the mixing condition for Glauber dynamics in \cite{HSZ19}.

\begin{theorem}\label{thm:sym-main-liner}
	 Let $\+G$ collect all $k$-uniform $d$-regular linear hypergraphs such that
    \[
		d \leq \frac{(1-\eps)}{4} \cdot \frac{2^k}{k^2}
    \]
	for some constant $\eps\in(0,1)$. There exists an algorithm such that given as input a linear hypergraph $\Phi\in\+G$, outputs a hypergraph independent set of $\Phi$ uniformly at random, in expected time $O \left(-\frac{1}{\log(1-\eps)}\cdot k^2d^2n \cdot( \log n+\log d) \right)$.
\end{theorem}

\subsection*{Relation to sampling solutions of general CNF formulas} 

The problem of sampling hypergraph independent sets is a special case of sampling the solutions of general CNF formulas or CSP instances which draw a lot of recent attention \cite{Moitra19,FGYZ21,JPV21,HSW21,JPV20,GLLZ19,GJL19,GGGY21,FHY21,FGW22,HWY22,GGGH,CMM,HWY22r}. For general $k$-CNF formulas where each variable is of degree $d$ (corresponding to the $k$-uniform $d$-regular hypergraphs here), the best condition needed for an efficient sampler is $d<\frac{1}{\-{poly}(k)}\cdot 2^{\frac{k}{40/7}}$~\cite{JPV21,HSW21} which is much worse than the condition $d<c\cdot 2^{\frac{k}{2}}$ for hypergraph independent sets. A main reason is that for general CNF formulas, the state space of the local Markov chain is no longer connected, and therefore one needs to project the chain onto a sub-instance induced by some special set of variables. The loss of the projection step, however, is not well-understood. On the other hand, this is not an issue for hypergraph independent sets, and therefore (almost) optimal bounds can be obtained.

We remark that the technique developed here for non-uniform graphs can also be applied to general CNF formulas to analyze the projection chain.  

\section{Preliminaries}\label{sec:prelim}
\subsection{Hypergraph independent sets}
    Recall in the introduction we mentioned that a set $S\subseteq V$ is an independent set of a hypergraph $\Phi=(V,\C)$ if $S\cap C\ne C$ for every $C\in \C$. Sometimes we represent it as a (binary) coloring $\sigma \in \set{0,1}^V$, which is the indicator for $S$. That is, we say $\sigma$ is an hypergraph independent set if $\forall C\in \C, \exists v\in C: \sigma(v)=0$. Denoting by $\Omega_{\Phi}$ the collection of all independent sets of $\Phi$, our goal is to efficiently produce a sample from the uniform measure $\mu$ on $\Omega_{\Phi}$.
    
    Let us fix some notations used throughout our discussion:
    \begin{itemize}
        \item We assume $V := [n]$ and $m := \abs{\C}$. We denote the degree of a vertex $v \in V$ as $d_v := \abs{ \set{C \in \C: v \in C} }$ and the maximum degree as $d := \max_{v\in V} d_v$. A hypergraph is $d$-regular if every vertex is of degree $d$ and is $k$-uniform if every hyperedge is of size $k$.
        
        \item For any $C \in \C$, we use $\Gamma^{+}_{\Phi^2}(C)$ to denote the set of hyperedges $C'$ such that either $ C \cap  C' \neq \emptyset$ or there exists $C^{*}$ such that $ C \cap  C^{*} \neq \emptyset$ and $ C^{*} \cap  C' \neq \emptyset$. Furthermore, let $\Gamma_{\Phi^2}(C):= \Gamma^{+}_{\Phi^2}(C) \setminus \set{C}$.
        
        \item For any coloring $\sigma \in \set{0,1}^V$, we use $\sigma^{v\gets r}$ to denote the coloring obtained after recoloring $v \in V$ by $r \in \set{0,1}$ in $\sigma$.
    \end{itemize}

\subsection{Systematic scan and coupling from the past}

    Fix a hypergraph $\Phi=(V,\C)$ and the uniform distribution $\mu$ over $\Omega_\Phi$. We use the so-called \emph{systematic scan} Markov chain to sample independent sets from $\mu$.
    
    Let us define a transition map $f: \Omega_\Phi \times V \times \set{0,1} \to \Omega_\Phi$ by
    \[
        f(\sigma; v, r) :=
		\begin{cases}
            \sigma^{v \gets r} & \mbox{if } \sigma^{v \gets r} \in \Omega_\Phi; \\
            \sigma  & \mbox{otherwise}.
        \end{cases}
    \]
    It takes the given coloring $\sigma \in \Omega_\Phi$ and tries to recolor vertex $v$ with the proposed $r$. Next, we fix a deterministic \emph{scan sequence} $v_1, v_2, \dots$ where $v_i := (i \bmod n) + 1$ and write
    \[
        F(\sigma;r_1,\dots,r_t):=f(f(\cdots f(\sigma;v_1,r_1)\cdots);v_t,r_t).
    \]
    Basically, it begins with the given coloring $\sigma \in \Omega_\Phi$ and runs $f$ for $t$ steps to update vertex colors. The vertices are updated periodically as specified by the scan sequence; the proposed colors for every step are provided by the arguments $r_1,\dots, r_t$. For a collection of initial states $S\subseteq \Omega_{\Phi}$, we use $F(S;r_1,\dots,r_t)$ to denote the set $\set{F(\sigma;r_1,\dots,r_t)\cmid \sigma\in S}$.

    \begin{definition}
        The \emph{systematic scan} is a Markov chain $(X_t)$ defined by
        \[ X_t := F(\sigma; R_1, \dots, R_t) \quad \forall t \]
        for some $\sigma \in \Omega_\Phi$ and some independent $\text{Bern}(1/2)$ variables $R_1, R_2, \dots$.
    \end{definition}

    The systematic scan is not a time-homogeneous Markov chain. However, we can (and will) bundle every $n$ steps into an atomic round so that the bundled version -- denoting its transition matrix $P_\Phi$ -- is homogeneous. It is easy to check that $P_{\Phi}$ is irreducible and aperiodic with stationary distribution $\mu$. 

    \bigskip
    Given a Markov chain with stationary $\mu$, in the usual Markov chain Monte Carlo method, one obtains a sampler for $\mu$ as follows: Starting from some initial $X_0$, simulate the chain for $t$ steps with sufficiently large $t$ and output $X_t$. The fundamental theorem of Markov chains says that if the chain is finite, irreducible, and aperiodic, then the distribution of $X_t$ converges to $\mu$ when $t$ approaches infinity. However, since we always terminate the simulation after some fixed finite steps, the sampler obtained is always an ``approximate'' sampler instead of a ``perfect'' one.

    The work of \cite{PW96} proposed an ingenious method called \emph{coupling from the past} (CFTP) to simulate the given chain in a reverse way with a random stopping time. A perfect sampler for $\mu$ can be obtained in this way. Roughly speaking, it essentially simulates an infinitely long chain using only finitely many steps. To achieve this, it relies upon a routine to detect whether the (finite) simulation coalesces with the virtual infinite chain. Once the coalescence happens, one can output the result of the simulation -- which is also the result of the infinite chain and thus follows the stationary distribution perfectly.
    
    We use CFTP to simulate our systematic scan and obtain a desired perfect sampler for hypergraph independent sets. The key ingredient to apply CFTP is how to detect coalescence. We describe and analyze our algorithm in \Cref{sec:sampler}.

\section{Perfect Sampler via Information Percolation}\label{sec:sampler}

In this section, we describe our perfect sampler for hypergraph independent sets. With the help of a data structure introduced in \cite{HSW21}, we apply the argument of information percolation to analyze the algorithm. We utilize the monotonicity of the hypergraph independent sets and establish a sufficient condition for the algorithm to terminate.

\subsection{The witness graph}
\label{sec:witness-graph}
In this section, we introduce the notion of witness graph $H_T=(V_T, E_T)$ for the systematic scan up to some time $T \in \^N$. A similar structure was used in \cite{HSW21} for sampling general CNF formulas.

Given a vertex $v \in V$, we denote its last update time up to moment $t$ as
\[
 \!{UpdTime}(v,t) = \max\set{t^*\leq t:  v_{t^*}=v}.
\]
Clearly $\!{UpdTime}(v,t) \in (t-n,t]$ is a deterministic number.

For $C \in \C$ and $t \in [T]$, let $e_{C,t} := \set{ \!{UpdTime}(v,t) : v \in C }$ be timestamps when the elements in $C$ got their latest updates up to time $t$. Conversely, we say $e_{C,t}$ has \emph{label} $C$ and denote it by $C(e_{C,t}) := C$. For convenience, the parenthesis notation $e_{C,t}=(t_1,t_2,\dots,t_{\abs{C}})$ means that $e_{C,t}$ is represented in time order $t_1 < t_2 < \cdots < t_{\abs{C}}$.

The vertex set of the witness graph is given by
\[
    V_T := \set{ e_{C,t} : t \in [T], v_t \in C \in \C}.
\]
We put a directed edge $e_{C,t} \to e_{C',t'}$ into $E_T$ if and only if $t'<t$ and $(e_{C,t} \cap e_{C',t'})\neq \emptyset$. Note that $H_T$ is acyclic, since $\max e > \max e'$ for any directed edge $e \to e'$. We remark that the witness graph is a deterministic object that does not incorporate any randomness. 

\begin{figure}[tbp]
    \centering
    \includegraphics[scale=0.9]{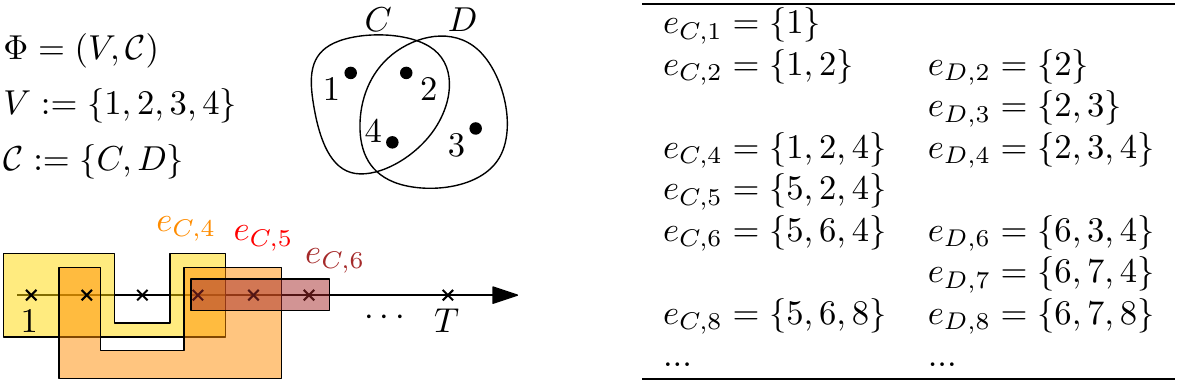}
    \caption{An illustration of the witness graph. In this example, the hypergraph $\Phi$ has four vertices and two hyperedges. We list the vertices of the witness graph in the table on the right. The lower left picture visualizes three vertices in $H_T$; all of them have label $C$.}
    \label{fig:witness-graph}
\end{figure}

The following lemmas measure the number of vertices labeled by a certain $C\in \C$ that are 2-distant from a given vertex in the witness graph. It is useful throughout the enumeration in later discussions.

\begin{lemma}\label{branchingnumber}
    For any $e_{C,t}\in V_T$ and $C'\in \C$, we have \[
        \abs{\set{e_{C',t'} \in V_T: \mathrm{dist}(e_{C,t},e_{C',t'})=2}}\leq 2\abs{C'}.
    \]
\end{lemma}
\begin{proof}
    If $C'\notin \Gamma^+_{\Phi^2}(C)$ then $\abs{ \set{e_{C',t'} \in V_T: \mathrm{dist}(e_{C,t},e_{C',t'})=2}}=0$. Otherwise, $t'\in [t-2n+2,t)$, and there are at most $2\abs{C'}$ timestamps $e_{C',t'}$ satisfying $t'\in [t-2n+2,t)$.
\end{proof}

% \begin{lemma}\label{branchingnumberII}
%     For any $e_{C,t}\in V_T$ and $C'\in \C$, we have 
%     \[
%         \abs{\set{e_{C',t'} \in V_T: \mathrm{dist}(e_{C,t},e_{C',t'})=2,e_{C,t}\cap e_{C',t'}=\emptyset,C'\cap C\neq \emptyset}}\leq d\abs{C}.
%     \]
% \end{lemma}
% \begin{proof}
%     Let $e_{C,t}=(t_1,t_2,\dots,t_k)$ where $t_k=t$,
%     for any $e_{C',t'} \in V_T$ such that $\mathrm{dist}(e_{C,t},e_{C',t'})=2$, $e_{C,t}\cap e_{C',t'}=\emptyset$, and $C'\cap C\neq \emptyset$, we have $t'\in [t-2n+2,t_1-1]$ and $v_{t'}\in C$.
%   There are at most $\abs{C}$ updating times $t'$ such that $t'\in [t-2n+2,t_1-1]$ and $v_{t'}\in C$.  $e_{C',t'}$ satisfying $t'\in [t-2n+2,t)$.
% \end{proof}

\subsection{Coalescence and percolation}
For any $L \in \^N$ we fix $T = T(L) := n(L+1)$. On top of the witness graph $H_T$ we define a probability space as follows. We tie an independent $\text{Bern}(1/2)$ variable $B_t$ to each time point $t \in [T]$. We say a vertex $e \in V_T$ is \emph{open} if $B_t=1$ for all $t\in e$, and call a set of vertices $P\subseteq V_T$ \emph{open} if all $e\in P$ are open. The event $\B_L$ is defined as
\begin{quote}
    ``$H_T$ contains an induced open path $P = (e_1, \dots, e_L) \subseteq V_T$ of length $L$ where $e_1 \cap (T-n,T]\ne\emptyset$''.
\end{quote}

For the coming few pages, the notation $\B_L(R_1, \dots, R_T)$ indicates that we are using (external) random variables $R_1, \dots, R_T$ as \emph{concrete realizations} of our abstract variables $B_1, \dots, B_T$. Needless to say, $R_1, \dots, R_T$ themselves should be independent $\text{Bern}(1/2)$ for such notation to make sense.

Starting from \Cref{sec:percolation}, however, we will switch back to the abstract setting and sweep the concrete realization under the rug.

Recall $F$ is the transition map for our systematic scan.  
\begin{lemma}\label{lem:detection} 
    Assume $L \in \^N$, $T := n(L+1)$ and $R_1, \dots, R_T$ are independent $\text{Bern}(1/2)$ variables. If $\abs{F(\Omega_\Phi; R_1,\dots,R_T)} > 1$ then $\B_L = \B_L(R_1, \dots, R_T)$ happens. 
\end{lemma}

To prove \Cref{lem:detection} we specify a \emph{grand coupling}, namely a family of Markov chains that share the same random sequence $R_1, \dots, R_T$ provided by the lemma.

For every $\sigma \in \Omega_\Phi$, define a copy of systematic scan $(X_{\sigma,t})_{0 \leq t \leq T}$ by
\[
    X_{\sigma,t} := F(\sigma, R_1, \dots, R_t) \quad \forall 0 \leq t \leq T.
\]

In addition, we define an auxiliary Markov chain $(Y_t)_{0 \leq t \leq T}$ by
\[
    Y_0 := \set{1}^V, \qquad Y_t := Y_{t-1}^{v_t \gets R_t} \quad \forall 1 \leq t \leq T.
\]

The chain $(Y_t)_{0\leq t\leq T}$ dominates the execution of $(X_{\sigma,t})_{0\leq t\leq T}$ by monotonicity:
\begin{proposition}\label{monotonicity}
    For all $\sigma\in \Omega_\Phi$ and $0 \leq t \leq T$, we have $X_{\sigma,t}\leq Y_{t}$.
\end{proposition}
\begin{proof}
    Initially, $X_{\sigma,0} = \sigma \leq \set{1}^V = Y_0$ for all $\sigma\in \Omega_\Phi$. At any time $t \geq 1$, all the chains update the same vertex $v_t$ and (i) if $R_t=0$ then $X_{\sigma,t}(v_t) = 0$; (ii) if $R_t=1$ then $Y_t(v_t)=1$. So the ordering $X_{\sigma,t} \leq Y_t$ is preserved throughout.
\end{proof}
    
\begin{proposition}\label{percolation}
    Let $0 \leq t \leq T$ be a time point. If there exist $\sigma,\tau: {X_{\sigma,t}(v_t)\neq X_{\tau,t}(v_t)}$, then there is a hyperedge $C\in \C$ containing $v_t$ such that $C\setminus \set{v_t}$ was fully colored by ``1'' in exactly one of $X_{\sigma,t}$ and $X_{\tau,t}$. Furthermore, $Y_{t}(C)=\set{1}^{C}$.
\end{proposition}
\begin{proof}
    Since ${X_{\sigma,t}(v_t) \neq X_{\tau,t}(v_t)}$, we know that $R_t=1$ and thus $Y_t(v_t)=1$. According to the definition of transition map $f$, one of the two chains failed to recolor $v_t$ by ``1'' because there exists a hyperedge $C$ such that $C\setminus \set{v_t}$ was fully colored ``1'' in that chain just before time $t$. Hence $Y_t(C\setminus\set{v_t})=\set{1}^{C \setminus \set{v_t}}$ by \Cref{monotonicity}. Putting together we have $Y_t(C)=\set{1}^C$.
\end{proof}

We are now ready to prove \Cref{lem:detection}.

\begin{proof}[Proof of \Cref{lem:detection}]
    Assume $\abs{F(\Omega_\Phi; R_1,\dots,R_T)} > 1$; that is, there exist $\sigma,\tau \in \Omega_\Phi$ and $v\in V$ such that $X_{\sigma,T}(v)\neq X_{\tau,T}(v)$. We inductively construct a path in $H_T$ as follows:
    \begin{enumerate}
        \item Set $i \gets 1$. Let $t_1 := \!{UpdTime}(v,T) \in (T-n, T]$.
        \item While $t_i \geq n$ do the following. Regarding \Cref{percolation} and the fact $X_{\sigma,t_i}(v_{t_i}) \neq X_{\tau,t_i}(v_{t_i})$, there is a hyperedge $C_i \in \C$ containing $v_{t_i}$ such that $C_i\setminus \set{v_{t_i}}$ was fully colored by ``1'' in exactly one of $X_{\sigma,t_i}$ and $X_{\tau,t_i}$. So we may find an earliest time $t_{i+1} \in e_i := e_{C_i,t_i}$ such that $X_{\sigma,t_{i+1}}(v_{t_{i+1}})\neq X_{\tau,t_{i+1}}(v_{t_{i+1}})$ and $t_{i+1} < t_i$. Moreover, $Y_{t_i}(C_i)=\set{1}^{C_i}$ implies that $e_i$ is open. Let $i \gets i+1$ and repeat.
    \end{enumerate}
    
    Note that $t_{i+1} \in e_i$ by definition. On the other hand $t_{i+1} \in e_{i+1} = e_{C_{i+1},t_{i+1}}$ since $v_{t_{i+1}} \in C_{i+1}$. Combining with the condition $t_{i+1} < t_i$, we see $t_{i+1} \in (e_i \cap e_{i+1}) \setminus \set{t_i}$ and consequently $e_i \to e_{i+1}$ is indeed an edge in $H_T$. Therefore the above procedure returns an open (not necessarily induced) path $P^\circ = (e_1, \dots, e_r) \subseteq V_T$ where $r$ is the number of rounds. It starts at vertex $e_1$ that intersects $(T-n,T]$, and ends at vertex $e_r$ that intersects $[1,n)$.
    
    Let $P \subseteq P^\circ$ be a \textit{shortest path} from $e_1$ to $e_r$ in $H_T[P^\circ]$. It must be an induced path in $H_T$ since $H_T$ is acyclic. It is open since $P \subseteq P^\circ$. Finally, every vertex $e \in P$ ``spans'' a time interval of at most $n$, so the length of $P$ is at least $T/n - 1 = L$.
\end{proof}

\subsection{The perfect sampler}
    We are ready to introduce our perfect sampler for hypergraph independent sets. Let $L \in \^N$ be a parameter to be fixed later (roughly in the order of $O(\log m)$) and $T = T(L) := n(L+1)$ as usual. The algorithm follows the standard framework of CFTP.

    \begin{algorithm}[th]
	    \caption{The CFTP sampler of hypergraph independent sets}
	    \label{algo:CFTP}
	    \KwIn{A hypergraph $\Phi=(V,\C)$.}
        \KwOut{An independent set of $\Phi$ sampled uniformly from $\mu$.}
		$j \gets 0$\;
		\Repeat{$\!{Detect}(\Phi,R_{-jT+1},\dots,R_{-(j-1)T})$}{
		    $j \gets j + 1$\;
		    generate independent $\text{Bern}(1/2)$ variables $R_{-jT+1}, \dots, R_{-(j-1)T}$\;
		}
		\Return{$F \left( \set{0}^V; R_{-jT+1}, \dots, R_0\right)$.}
	\end{algorithm}
	
	\begin{algorithm}[H]
	    \caption{The Detect subroutine}
	    \label{algo:detect}
	    \KwIn{A hypergraph $\Phi=(V,\C)$ and a sequence $R_1,\dots,R_T$.}
        \KwOut{Whether $\B_L(R_1, \dots, R_T)$ happens.}
        \ForEach{intersecting $C,C' \in \C$}{
            \ForEach{$t \in [T]: v_t \in C$}{
                \ForEach{$t' \in [t-n,t): v_{t'} \in C'$}{
                    connect $e_{C,t} \to e_{C',t'}$ if they intersect and are open with respect to $R_1, \dots, R_T$\;
                }
            }
        }
		breadth-first search from all $e_{C,t}: C \in \C, t \in (T-n, T]$\;
		\Return{if the path has length $\geq L$.}
    \end{algorithm}

    \begin{theorem}\label{thm:cftp}
        Suppose there exist constant $\eps \in (0,1)$ and $\beta$ such that $\Pr{\B_L}\le \beta\cdot(1-\eps)^L$ for all $L$. With the concrete choice $L := \left\lceil \frac{\log(2\beta)}{-\log(1-\eps)} \right\rceil$, \Cref{algo:CFTP} terminates with probability $1$ and has expected running time $O \left( \frac{\log(2\beta)}{-\log(1-\eps)} \cdot \sum_{C \cap C' \neq \emptyset} \abs{C'}  \abs{C} \right) $. Moreover, its output distribution $\nu$ is exactly $\mu$ upon termination.
    \end{theorem}
    %= O(m k^2 d \cdot \frac{\log(2\beta)}{-\log(1-\eps)})

    \begin{proof}
        The $\!{Detect}$ subroutine (\Cref{algo:detect}) essentially constructs the witness graph $H_{T}$ and  decides if  if $\B_L(R_1, \dots, R_T)$ happens. Since the inputs fed by \Cref{algo:CFTP} are always independent $\text{Bern}(1/2)$ variables as required, we have
        \[
            p := \Pr{\!{Detect}(\Phi,R_{-jT+1},\dots,R_{(-j-1)T})} = \Pr{\B_L} \leq \beta \cdot (1-\eps)^L = \frac{1}{2}.
        \]
        where the inequality follows from the assumption of the theorem and \Cref{lem:detection}.
        
        Note that \Cref{algo:CFTP} feeds disjoint (hence independent) sequences into \Cref{algo:detect} in different rounds, so the corresponding return values are independent. Therefore, the total number of rounds $J$ follows geometric distribution $\text{Geom}(1-p)$. In particular $\Pr{J < \infty} = 1$ and $\E{J} = \frac{1}{1-p} \leq 2$. So the algorithm terminates with probability $1$.
        
        Next we analyze the running time of a single call to \Cref{algo:detect}. For each $C,C'$ pair, there are $L \cdot \abs{C}$ valid choices of $t$ and $\abs{C'}$ valid choices of $t'$. Each iteration of the innermost loop can be implemented in constant time, as we may index a vertex $e_{C,t}$ by the pair $(C,t)$ to allow random access and test its openness by looking up a precomputed table. So the nested loop takes
         $O \left( L \cdot \sum_{C \cap C' \neq \emptyset} \abs{C} \cdot  \abs{C'} \right)$
        time to finish.
        The standard breadth-first search consumes time proportional to the number of edges in the witness graph, which is bounded by the same quantity.
        
        Finally, recall that \Cref{algo:CFTP} calls \Cref{algo:detect} $J$ times. But $\E{J} \leq 2$, so the expected running time of our sampler is  $O \left( L \cdot \sum_{C \cap C' \neq \emptyset} \abs{C} \cdot  \abs{C'} \right)$.
        
        We now proceed to show $\nu = \mu$, namely the algorithm outputs perfect sample. Upon termination, the percolation event $\B_L(R_{-JT+1}, \dots, R_{-(J-1)T})$ does not happen due to the loop condition. Hence $\abs{F(\Omega_\Phi; R_{-JT+1}, \dots, R_{-(J-1)T})} = 1$ by \Cref{lem:detection}. But this in particular means
        \begin{equation}
            \abs{F(\Omega_\Phi; R_{-JT+1}, \dots, R_0)} = 1.
            \label{eq:coalescence}
        \end{equation}
        
		For each integer $s \geq 0$, we define a (virtual) copy $\left(Z^{s}_t\right)_{-sn \leq t \leq 0}$ of systematic scan which starts from $\set{0}^V$ at time $-sn$:
		\[
		    Z^s_t := F \left( \set{0}^V; R_{-sn+1}, \dots, R_t \right).
	    \]
		Let $\nu^{s}$ denote the distribution of $Z^{s}_0$. Then $\lim_{s \to \infty} \nu^s = \mu$ by convergence theorem. On the other hand,
		\begin{align*}
			&\Vert \nu^{s} - \nu \rVert \\
			&\leq \Pr{ F \left(\set{0}^V; R_{-sn+1},\dots,R_0 \right) \neq F\left(\set{0}^V; R_{-JT+1},\dots,R_0 \right) }  \tag{coupling}\\
			&\leq \Pr{ sn < JT }  \tag{by (\ref{eq:coalescence})} \\
			&\to 0.  \tag{as $s \to \infty$}
		\end{align*}
		Hence $\nu = \lim_{s \to \infty} \nu^s = \mu$.
    \end{proof}

\subsection{Proofs of \Cref{thm:asym-main,thm:sym-main,thm:sym-main-liner} }

    Armed with \Cref{thm:cftp}, we only need to bound $\Pr{\+B_L}$ under various conditions. We bound the quantity in \Cref{sec:percolation}. As a result, \Cref{thm:asym-main} follows from \Cref{lem:percolation-asym}, \Cref{thm:sym-main} follows from \Cref{lem:percolation-sym}, and  \Cref{thm:sym-main-liner} follows from \Cref{lem:percolation-sym-linear}.

\section{Percolation on Witness Graphs}
\label{sec:percolation}
Let $\Phi=(V,\C)$ be a hypergraph with $\abs{V}=n$ and let $H_T$ be its witness graph for any $T = (L+1) n$, $L \in \^N$. We will analyze the probability $\Pr{\B_L}$ in the abstract probability space, where each time point $t \in [T]$ is associated with an independent $\text{Bern}(1/2)$ variable.

\subsection{General hypergraphs}
\label{sec:asym}
\begin{lemma}\label{lem:percolation-asym}
	If there exists a constant $\eps \in (0,1)$ and a function $x:\C\to(0,1)$ satisfying
	\[
		\forall C\in \C: 2\abs{C}\cdot 2^{-\abs{C}}\le (1-\eps) \cdot x(C)\cdot \prod_{C'\in\Gamma_{\Phi^2}(C)}(1-x(C')),
	\]
	then $\Pr{\B_L}\le \sum_{C\in\C}\frac{x(C)}{1-x(C)}\cdot (1-\eps)^{\lfloor\frac{L+1}{2}\rfloor}$ for all $L \in \^N$.
\end{lemma}

For every $C\in \C$, we use $\+P_{C,L}$ to denote the collection of induced paths $(e_1,e_2,\dots,e_L)$ in $H_T$ where $e_1\cap (T-n,T] \ne \varnothing$ and the label of $e_1$ is $C$ (namely $C(e_1) = C$). Let 
\[
	\+P^2_{C,\lfloor\frac{L+1}{2}\rfloor} :=\set{(e_1,e_3,e_5\dots, e_{2\lfloor\frac{L-1}{2}\rfloor+1})\cmid (e_1,e_2,e_3\dots,e_L)\in \+P_C}.
\]
Then by the union bound, we have
\begin{align}
\Pr{\B_L} 
\notag
&\le \sum_{C\in \C}\sum_{P\in \+P_{C,L}}\Pr{P\mbox{ is open}}\le \sum_{C\in \C}\sum_{P\in \+P^2_{C,\lfloor\frac{L+1}{2}\rfloor}}\Pr{P\mbox{ is open}} \\
\label{eqn:p2prob}
&=\sum_{C\in \C}\sum_{(e_1,\dots,e_{\lfloor\frac{L+1}{2}\rfloor})\in \+P^2_{C,\lfloor\frac{L+1}{2}\rfloor}}\prod_{i=1}^{\lfloor\frac{L+1}{2}\rfloor}2^{-\abs{C(e_i)}},
\end{align}
where the last equality follows from the fact that $e_i\cap e_{i+1}=\emptyset$ holds for every pair of consecutive vertices of a path in $\+P^2_{C,\lfloor\frac{L+1}{2}\rfloor}$.

\begin{lemma}\label{lem:p2prob}
If there exists a function $x:\C\to(0,1)$ satisfying
	\[
		\forall C\in \C: 2\abs{C}\cdot 2^{-\abs{C}}\le (1-\eps) \cdot x(C)\cdot \prod_{C'\in\Gamma_{\Phi^2}(C)}(1-x(C')),
	\]
then 
	\[
	\sum_{C\in \C}\sum_{(e_1,\dots,e_{\lfloor\frac{L+1}{2}\rfloor})\in \+P^2_{C,\lfloor\frac{L+1}{2}\rfloor}}\prod_{i=1}^{\lfloor\frac{L+1}{2}\rfloor}2^{-\abs{C(e_i)}}\le \sum_{C\in\C}\frac{x(C)}{1-x(C)}\cdot (1-\eps)^{\lfloor\frac{L+1}{2}\rfloor}.
	\]
\end{lemma}

\Cref{lem:percolation-asym} is clearly a consequence of \Cref{lem:p2prob}.  We prove the latter by analyzing a \emph{multi-type Galton-Watson branching process}.

\subsubsection{Multi-type Galton-Watson branching process}

Recall $\Phi=(V,\C)$ is a fixed hypergraph and we assume that there exists a function $x:\C\to(0,1)$ assigning each $C\in\C$ a number in $(0,1)$. A $\C$-labelled tree is a tuple $\tau=(V_\tau,E_\tau,\+L)$ where $\+L:V_\tau\to\C$ labels each vertex in $V_\tau$ with a hyperedge in $\C$. 

Let $C\in \C$ be a hyperedge. Consider the following process which generates a random $\C$-labelled tree with the root labelled with $C$:
\begin{itemize}
	\item First produce a root vertex $u$ with label $C$. Initialize the active set $A$ as $\set{u}$.
	\item Repeat the following until $A$ is empty:
		\begin{itemize}
			\item Pick some $u\in A$.
			\item For every $C'\in \Gamma^{+}_{\Phi^2}(\+L(u))$: create a new child for $u$ labelled with $C'$ \emph{with probability} $x(C')$ independently; Add the new child to $A$.
			\item Remove $u$ from $A$.
		\end{itemize}
\end{itemize}

We let $\+T_C$ be the set of all labelled trees that can be generated by the above process and let $\mu_{\+T_C}$ be the distribution over $\+T_C$ induced by the process.

\begin{lemma}\label{lem:GWprob}
	For every labelled tree $\tau=(V_\tau,E_\tau,\+L)\in\+T_C$, it holds that
	\[
		\mu_{\+T_C}(\tau)= \frac{1-x(C)}{x(C)}\cdot \prod_{v\in V_\tau} x(\+L(v))\cdot\prod_{C'\in\Gamma_{\Phi^2}(\+L(v))}(1-x(C')).
	\]
\end{lemma}
\begin{proof}
    For a vertex $v\in V_\tau$ we use $W_v \subseteq  \Gamma^+_{\Phi^2}({\+L}(v))$ to denote the set of labels that do not occur as a label of some child nodes of $v$. Then clearly, the probability that the Galton-Watson branching process produces exactly the tree $\tau$ is given by
    \[
       \mu_{\+T_C}(\tau) = \frac{1}{x(C)} \prod_{v\in V_\tau} x({\+L}(v)) \prod_{C' \in W_v}(1-x(C')).
    \]
    In order to get rid of the $W_v$, we can rewrite the expression as 
    \begin{align*}
        \mu_{\+T_C}(\tau) &= \frac{1-x(C)}{x(C)} \prod_{v\in V_\tau} \frac{x({\+L}(v))}{1-x({\+L}(v))}\prod_{C' \in \Gamma_{\Phi^2}^{+}({\+L}(v))}(1-x(C'))\\
         &=\frac{1-x(C)}{x(C)} \prod_{v\in V_\tau}x({\+L}(v)) \prod_{C' \in \Gamma_{\Phi^2}({\+L}(v))}(1-x(C')).
    \end{align*}    
\end{proof}

For every $\ell\ge 1$, we also use $\+T_{C,\ell}$ to denote the set of labelled trees in $\+T_C$ containing exactly $\ell$ vertices. In the following, we slightly abuse notation and say that a labelled tree $\tau=(V_\tau,E_\tau,\+L)\in \+T_C$ (or $\+T_{C,\ell}$) if there exists some $\tau'\in \+T_C$ (or $\+T_{C,\ell}$) such that $\tau$ and $\tau'$ are isomorphic.

\subsubsection{Proof of \Cref{lem:p2prob}}

Let $\ell=\lfloor\frac{L+1}{2}\rfloor$. It follows from the assumption that for every $C\in\C$, 
\begin{align}
\notag	&\sum_{P=(e_1,e_2,\dots,e_{\ell})\in \+P^2_{C,\ell}}\prod_{i=1}^{\ell}2^{-\abs{C(e_i)}} \\
	&\quad\le (1-\eps)^{\ell}\cdot \sum_{P=(e_1,e_2,\dots,e_{\ell})\in \+P^2_{C,\ell}}\prod_{i=1}^{\ell}\frac{1}{2\abs{C(e_i)}}\cdot x(C(e_i))\prod_{C'\in\Gamma_{\Phi^2}(C(e_i))}(1-x(C')).
\end{align}

We now define a mapping $\Psi:\+P^2_{C,\ell}  \to \+T_{C,\ell}$ that maps each $P=(e_1,e_2,\dots,e_\ell)$ to a labelled directed path $\tau=(V_\tau,E_\tau,\+L)$ where 
\begin{itemize}
	\item $V_\tau = \set{u_1,u_2,\dots,u_\ell}$;
	\item $E_{\tau} = \set{(u_i,u_{i+1})\cmid i\in[\ell-1]}$; and 
	\item for every $i\in[\ell]$, $\+L(u_i) = C(e_i)$.
\end{itemize}

The mapping $\Psi$ is not necessarily an injection, but we can bound its multiplicity.

\begin{lemma}\label{lem:multi-map}
	For every labelled tree $\tau = (V_{\tau},E_\tau,\+L)\in \+T_{C,\ell}$ with $V_\tau = \set{u_1,\dots,u_\ell}$ and $E_{\tau} = \set{(u_i,u_{i+1})\cmid i\in[\ell-1]}$, it holds that
	\[
	\abs{\Psi^{-1}(\tau)} \le 2^{\ell-1} \cdot \prod_{i=1}^\ell \abs{\+L(u_i)}.
	\]
\end{lemma}
\begin{proof}
	We prove it by induction:
	\begin{enumerate}
	    \item When $\ell =1$, the lemma holds because there are at most $\abs{C}$ timestamps $e$ with label $C$ such that $e\cap {(T-n+1,T]\ne \varnothing}$; 
	    \item We assume the statement holds for $\ell \geq 1$. Let $V_\tau=\set{u_1,u_2,\dots,u_\ell,u_{\ell+1}}$ and $V_{\tau'}=\set{u_1,u_2,\dots,u_\ell}$. According to the definition of $\Psi$, we know that for all $P\in \Psi^{-1}(\tau)$, it was extended from some $P'\in \Psi^{-1}(\tau')$. Therefore, it suffices to analyze the possible extension from $P'$ to $P\in \Psi^{-1}(\tau)$. Assuming $P'=(e_1,e_2,\dots,e_\ell)$ and $P=\set{e_1,e_2,\dots,e_\ell,e_{\ell+1}}$, if $P$ was extended from $P'$, then $\mathrm{dist}(e_\ell,e_{\ell+1})=2$ and $C(e_{\ell+1})=\+L(u_{\ell+1})$. Therefore, there are at most $2\abs{\+L(u_{\ell+1})}$ timestamps $e_{\ell+1}$ satisfying $P\in \Psi^{-1}(\tau)$ according to \Cref{branchingnumber}. Combining with the induction hypothesis, the statement holds for $\ell+1$.
	\end{enumerate}
\end{proof}
Clearly for those $\tau\in \+T_{C,\ell}$ which are not directed paths, its pre-image $\Psi^{-1}(\tau)=\emptyset$. Therefore, it follows from \Cref{lem:GWprob} and \Cref{lem:multi-map} that
\begin{align*}
	&\sum_{P=(e_1,e_2,\dots,e_{\ell})\in \+P^2_{C,\ell}}\prod_{i=1}^{\ell}\frac{1}{2\abs{C(e_i)}}\cdot x(C(e_i))\prod_{C'\in\Gamma_{\Phi^2}(C(e_i))}(1-x(C'))\\
	&\quad \le \sum_{\substack{\tau=(V_\tau,E_\tau,\+L)\in\+T_{C,\ell}:\\V_\tau=\set{u_1,\dots,u_\ell}}}\prod_{i=1}^{\ell} x(\+L(u_i))\prod_{C'\in\Gamma_{\Phi^2}(\+L(u_i))}(1-x(C'))\\
	&\quad =\frac{x(C)}{1-x(C)}\sum_{\tau\in\+T_{C,\ell}}\mu_{\+T_C}(\tau)\\
	&\quad < \frac{x(C)}{1-x(C)}.
\end{align*}

\subsection{Refined analysis for uniform hypergraphs}\label{refined}

If the instance $\Phi=(V,\C)$ is a $d$-regular $k$-uniform hypergraph, we can choose $x(C) = \frac{1}{d^2k^2}$ 
 in \Cref{lem:percolation-asym} and the condition becomes $d\le \frac{c\cdot 2^{k/2}}{k^{1.5}}$ for some constant $c>0$. In this section, we present a refined analysis for this case which removes the denominator $k^{1.5}$.
 
  \begin{lemma}\label{lem:percolation-sym}
 	For all $\eps\in(0,1)$ and $k \geq 2$, if
	\[
		d \leq  \frac{(1-\eps)}{20}\cdot \frac{2^{k/2}}{k}
	\]
	then $\Pr{\B_L}\le \frac{k}{2^k}\cdot m \cdot (1-\eps)^{\lfloor\frac{L}{2}-1\rfloor}$ for all $L \in \^N$.
 \end{lemma}

 For every induced path $(u_1,\dots,u_\ell)$ in $H_T$, and $L>0$, we use $\+P_{(u_1,\dots,u_\ell),L}$ to denote the collection of induced paths in $H_T$ of length $L$ with prefix $(u_1,\dots,u_\ell)$\footnote{Unlike the notation $\+P_{C,L}$ defined in \Cref{sec:asym}, we do not require $u_1\cap \set{T-n+1,\dots,T}\ne\emptyset$ here.}.
 
 \begin{lemma}\label{lem:sym}
    For all $\eps\in(0,1)$ and $k \geq 2$, if
	\[
		d \leq  \frac{(1-\eps)}{20}\cdot \frac{2^{k/2}}{k}
	\]
 then for every $u\in V_{H_T}$ and every $L \in \^N$,
 \[
 	\Pr{\mbox{$\exists$open $P\in \+P_{(u),L}$} \mid u\mbox{ open}}\le (1-\eps)^{\lfloor\frac{L}{2}-1\rfloor}.
 \]
 \end{lemma}

 We first show that \Cref{lem:percolation-sym} follows from \Cref{lem:sym}.
 
 \begin{proof}[Proof of \Cref{lem:percolation-sym}]
    By the union bound, we have
    \begin{align*}
        \Pr{\B_L} 
        &\le \sum_{C\in \C}\sum_{e_{C,t}:t\in (T-n,T]}\Pr{\mbox{$\exists$open $P\in \+P_{(e_{C,t}),L}$}}\\
        &= \sum_{C\in \C}\sum_{e_{C,t}:t\in (T-n,T]}\Pr{\mbox{$\exists$open $P\in \+P_{(e_{C,t}),L}$} \mid e_{C,t} \mbox{ open}}\cdot \Pr{e_{C,t} \mbox{ open}}\\
        &\le k m (1-\eps)^{\lfloor\frac{L}{2}-1\rfloor} \cdot {2^{-k}}.
    \end{align*}
 \end{proof}

The remaining part of the section is devoted to a proof of \Cref{lem:sym}. We apply induction on $L$. The base case is that $L=1$ and $L=2$, in which the lemma trivially holds. For larger $L$ and every path $P=(u_1,u_2,u_3\dots,u_L)\in \+P_{(u_1),L}$, we discuss how the vertices $u_1$, $u_2$ and $u_3$ overlap.

Recall for every $u=e_{C,t}\in V_T$, $C(u)=C$ is its label. Similar to \cite{HSZ19}, we classify the tuple $(u_1,u_2,u_3)$ into three categories. 

\begin{itemize}
\item We say $u_2$ is \emph{good} if $\abs{C(u_2)\cap C(u_1)}\le \alpha\cdot k$ where $\alpha\in [0,1]$ is a parameter to be set;
\item If $u_2$ is not good, we say $(u_2,u_3)$ is of \emph{type I} if $C(u_3)\cap C(u_1)\ne\emptyset$;
\item If $u_2$ is not good, we say $(u_2,u_3)$ is of \emph{type II} if $C(u_3)\cap C(u_1)=\emptyset$.
\end{itemize}

Then we can write 
\begin{align}
	&\phantom{{}={}}\Pr{\mbox{$\exists$open $P\in \+P_{(u_1),L}$}\mid u_1\mbox{ open}}\label{eqn:3-cases}\\
	&\le \Pr{\mbox{$\exists$open $P=(u_1,u_2,u_3,\dots)\in \+P_{(u_1),L}:$ good $u_2$}\mid u_1\mbox{ open}}\notag\\
	&\quad+\Pr{\mbox{$\exists$open $P=(u_1,u_2,u_3,\dots)\in \+P_{(u_1),L}:$ type I $(u_2,u_3)$}\mid u_1\mbox{ open}}\notag\\
	&\quad+\Pr{\mbox{$\exists$open $P=(u_1,u_2,u_3,\dots)\in \+P_{(u_1),L}:$ type II $(u_2,u_3)$}\mid u_1\mbox{ open}}\notag
\end{align}

In the following, we bound the probabilities in the three cases respectively.

\subsection*{$u_2$ is good} In this case, we have $\abs{C(u_2)\cap C(u_1)}\le \alpha \cdot k$. Therefore, conditioned on $u_1$ being open, at least $(1-\alpha)\cdot k$ variables in $u_2$ are free (i.e. independent of those variables related to $u_1$). The number of choices of $C(u_2)$ is at most $k\cdot d$ and the number of choices of $u_2$ with fixed $C=C(u_2)$ is at most $k$. Combining this with the induction hypothesis, we have
\begin{align}
	&\phantom{{}={}}\Pr{\mbox{$\exists$open $P=(u_1,u_2,u_3,\dots)\in \+P_{(u_1),L}:$ good $u_2$}\mid u_1\mbox{ open}}\notag\\
	&=\Pr{\bigcup_{\mbox{\scriptsize good $u_2^*$}}\left\{\mbox{$\exists$open $P\in \+P_{(u_1,u_2^*),L}$}\right\}\mid u_1\mbox{ open}}\notag\\
	&\le \sum_{\mbox{\scriptsize good $u_2^*$}}\Pr{\mbox{$\exists$open $P\in \+P_{(u_1,u_2^*),L}$}\mid u_1\mbox{ open}}\notag\\
	&=\sum_{\mbox{\scriptsize good $u_2^*$}} \Pr{u_2^*\mbox{ open}\mid u_1\mbox{ open}}\cdot \Pr{\mbox{$\exists$open $P\in \+P_{(u_1,u_2^*),L}$}\mid (u_1,u_2^*)\mbox{ open}}\notag\\
	&= \sum_{\mbox{\scriptsize good $u_2^*$}} \Pr{u_2^*\mbox{ open}\mid u_1\mbox{ open}} \cdot \Pr{\mbox{$\exists$open $P' \in \+P_{(u_2^*),L-1}$}  \mid u_2^*\mbox{ open}}\notag\\
	&\le k^2d\cdot 2^{-(1-\alpha)k}\cdot (1-\eps)^{\lfloor\frac{L-3}{2}\rfloor}.\label{eqn:good}
\end{align}

\subsection*{$(u_2,u_3)$ is of type I} In this case, we know $C(u_3)$ is adjacent to $C(u_1)$ in $\Phi$ and therefore the number of the choices of $C(u_3)$ is at most $kd$. The choice of $u_3$ with fixed $C=C(u_3)$ is at most $2k$ by \Cref{branchingnumber}. On the other hand, since $P$ is an induced path, $u_1\cap u_3=\emptyset$ and $\Pr{u_3\mbox{ is open}\mid u_1\mbox{ is open}}=2^{-k}$. Combining these facts with the induction hypothesis, we have
\begin{align}
	&\phantom{{}={}}\Pr{\mbox{$\exists$open $P=(u_1,u_2,u_3,\dots)\in \+P_{(u_1),L}:$ type I $(u_2,u_3)$}\mid u_1\mbox{ open}}\notag\\
	&\le\Pr{\bigcup_{\substack{u_3^*:~ \exists u_2\text{ s.t.}\\(u_2,u_3^*)\text{ type I}}}\left\{\mbox{$\exists$open $P' \in \+P_{(u_3^*),L-2}$} \right\}  \mid u_1\mbox{ open}}\notag\\
	&\le \sum_{\substack{u_3^*:~ \exists u_2\text{ s.t.}\\(u_2,u_3^*)\text{ type I}}}\Pr{\mbox{$\exists$open $P' \in \+P_{(u_3^*),L-2}$} \mid u_1\mbox{ open}}\notag\\
	&=\sum_{\substack{u_3^*:~ \exists u_2\text{ s.t.}\\(u_2,u_3^*)\text{ type I}}} \Pr{u_3^*\mbox{ open}\mid u_1\mbox{ open}} \cdot  \Pr{\mbox{$\exists$open $P' \in \+P_{(u_3^*),L-2}$}\mid (u_1,u_3^*)\mbox{ open}}\notag\\
	&=\sum_{\substack{u_3^*:~ \exists u_2\text{ s.t.}\\(u_2,u_3^*)\text{ type I}}} \Pr{u_3^*\mbox{ open}\mid u_1\mbox{ open}} \cdot  \Pr{\mbox{$\exists$open $P' \in \+P_{(u_3^*),L-2}$}\mid u_3^* \mbox{ open}}\notag\\
	&\le 2k^2d\cdot 2^{-k} \cdot (1-\eps)^{\lfloor\frac{L-4}{2}\rfloor}.\label{eqn:typeI}
\end{align}

\subsection*{$(u_2,u_3)$ is of type II} This case is more complicated. We will enumerate all type II pairs $(u_2,u_3)$ and bound the sum of probabilities that some path in $\+P_{(u_2,u_3),L-1}$ is open conditional on $u_1$ being open. Our plan is to first enumerate $(C_2,C_3)$ and then consider $(u_2,u_3)$ whose labels are exactly $C(u_2)=C_2$ and $C(u_3)=C_3$. 

Let us write $C_1 := C(u_1)$ and fix any pair $(C_2,C_3) \in \C^2$ such that $\abs{C_2\cap C_1}>\alpha\cdot k$ and $C_3\cap C_1=\emptyset$. We aim to bound the probability
\begin{equation*}
	\sum_{\substack{u_2: C(u_2)=C_2 \\ u_3: C(u_3)=C_3}}\Pr{\mbox{$\exists$open $P \in \+P_{(u_1,u_2,u_3),L}$}\mid u_1 \mbox{ open}}.
\end{equation*}
Recall the notations in \Cref{sec:witness-graph}; we represent the $u_2$, $u_3$ of interest as $u_2=e_{C_2,t_2}$ and $u_3=e_{C_3,t_3}$ respectively.  By the construction of the witness graph, $u_1\rightarrow u_2$ implies $\min u_1 \leq t_2 < \max u_1$. Hence $t_2$ is sandwiched in an interval of length $n$ and has at most $k$ valid choices (as we require $v_{t_2} \in C_2$). After fixing $t_2$, we similarly have at most $k$ valid choices for $t_3$. Altogether there are at most $k^2$ many pairs $(u_2,u_3)$.

Now denote $a:= \abs{C_2\cap C_1}$ and $b:= \abs{C_3\cap C_2}$. For each pair $(u_2,u_3)$, we have $\abs{u_2\cap u_1}\leq a$, $\abs{u_3\cap u_2}\leq b$ and
\begin{align*}
	&\phantom{{}={}}\Pr{\mbox{$\exists$open $P \in \+P_{(u_1,u_2,u_3),L}$}\mid u_1 \mbox{ open}}\\
	&= \Pr{(u_2, u_3)\mbox{ open}\mid u_1\mbox{ open}} \cdot \Pr{\mbox{$\exists$open $P \in \+P_{(u_1,u_2,u_3),L}$}\mid (u_1,u_2,u_3)\mbox{ open}}\\
	&\le \Pr{(u_2, u_3)\mbox{ open}\mid u_1\mbox{ open}} \cdot  \Pr{\mbox{$\exists$open $P' \in \+P_{(u_3),L-2}$}\mid u_3\mbox{ open}}\\
	&\leq 2^{a+b-2k} \cdot (1-\eps)^{\lfloor\frac{L-4}{2}\rfloor}.
\end{align*}
where the last line used induction hypothesis. Therefore,
\begin{equation*}
	\sum_{\substack{u_2: C(u_2)=C_2 \\ u_3: C(u_3)=C_3}}\Pr{\mbox{$\exists$open $P \in \+P_{(u_1,u_2,u_3),L}$}\mid u_1 \mbox{ open}} \leq k^2\cdot  2^{a+b-2k} (1-\eps)^{\lfloor\frac{L-4}{2}\rfloor}.
\end{equation*}

It remains to sum over all $(C_2,C_3)$ pairs. Since we know that $\abs{C_2\cap C_1}>\alpha\cdot k$, the number of choices of $C_2$ is at most $\frac{d\cdot k}{\alpha\cdot k} = \frac{d}{\alpha}$\footnote{To see this, consider the following way to enumerate all $C_2$ incident to $C_1$ in $\Phi$: First pick a vertex in $C_1$ and then pick one of its incident hyperedges. This way we enumerated (with repetitions) $k \cdot d$ hyperedges in total, and every hyperedge $C_2: \abs{C_2\cap C_1}>\alpha \cdot k$ is enumerated at least $\alpha\cdot k$ times.}. For every fixed $C_2$, we let $\+C = \set{C_3^{(1)},C_3^{(2)},\dots,C_3^{(m)}}$ be the collection of all possible $C_3$. Denote $a := \abs{C_2\cap C_1}$ as before and for every $i\in [m]$ let $b_i := \abs{C_3^{(i)} \cap C_2}$. The probability of interest can therefore be written as
\begin{align*}
	&\phantom{{}={}} \Pr{\mbox{$\exists$open $(u_1,u_2,u_3,\dots)\in \+P_{(u_1),L}$:~ type-II $(u_2,u_3)$,~ $C(u_2)=C_2$} \mid u_1\mbox{ open}} \\
	&\le \sum_{\substack{(u_1,u_2,u_3,\dots)\in \+P_{(u_1),L}:\\(u_2,u_3)\mbox{ \scriptsize type II},\\C(u_2)=C_2}}\Pr{\mbox{$\exists$open $P \in \+P_{(u_1,u_2,u_3),L}$}\mid u_1 \mbox{ open}}\\
	&\le k^2\cdot \tp{\sum_{i=1}^m 2^{a+b_i-2k}}\cdot (1-\eps)^{\lfloor\frac{L-4}{2}\rfloor}.
\end{align*}

To bound the term $\sum_{i=1}^m 2^{a+b_i-2k}$, let us list some properties of the numbers involved:
\begin{itemize}
    \item $a \le k$  and $1 \le m \le (k-a) d$;
	\item $1 \le b_i \le k-a$ for all $i\in [m]$;
	\item $\sum_{i=1}^m b_i\le (k-a)d$. (This can be seen by an argument similar to the last footnote and the fact that $C_3 \cap C_1=\emptyset$.) 
\end{itemize}
Roughly the numbers are acting against each other: If $a$ and $m$ are large then the $b_i$'s are small. So it is possible to control $\sum_{i=1}^m 2^{a+b_i-2k}$ reasonably:

\begin{lemma}\label{lem:opt}
    Suppose $a \leq k$ and $1 \le m \le (k-a) d$. Then for any integers $b_1, \dots, b_m \in [k-a]$ such that $\sum_{i=1}^m b_i \leq (k-a)d$, we have the inequality $\sum_{i=1}^m 2^{a+b_i-2k} \le d\cdot 2^{-k}$.
\end{lemma}

\begin{proof}
    For each $j \in [k-a]$ we introduce a counter $x_j := \abs{\set{i: b_i = j}}$ that counts the number of $b_i$'s taking a specific value $j$. Then the constraint of the lemma translates to
    \[
        \sum_{j=1}^{k-a} j \cdot x_j \leq (k-a)d,
    \]
    and we want to show
    \[
        \sum_{j=1}^{k-a} 2^{a+j-2k} \cdot x_j \leq d \cdot 2^{-k}.
    \]
    
    To this end, we consider a (relaxed) linear program with variables $x_1, \dots, x_{k-a}$:
    \begin{align*}
        \max \quad & \sum_{j=1}^{k-a} 2^{a+j-2k} \cdot x_j \\
        \mbox{s.t} \quad & \sum_{j=1}^{k-a} j \cdot x_j \leq (k-a)d\\
        &x_j\geq 0 \quad (\forall j \in [k-a]).
    \end{align*}
    
    By the strong duality theorem of linear programming, its maximum value equals the minimum value of the dual program
    \begin{align*}
        \min \quad & (k-a)d \cdot y \\
        \mbox{s.t} \quad & j \cdot y \geq 2^{a+j-2k} \quad (\forall j \in [k-a])\\
        &y \geq 0.
    \end{align*}
    
    Clearly the minimum value is obtained when $y$ is as small as possible. It is clear that the sequence $h_j := \frac{2^j}{j}$ is monotonically increasing for integers $j \geq 1$, so the minimum possible $y$ is given by the $(k-a)$-th constraint, namely
    \[
        y^* := \frac{2^{a+(k-a)-2k}}{k-a} = \frac{2^{-k}}{k-a}.
    \]
    Therefore, the minimum value of the dual -- also the maximum value of the primal -- is exactly $d \cdot 2^{-k}$, as desired.
\end{proof}

Returning to the previous discussion, we derive
\begin{align}
	&\phantom{{}={}}\Pr{\mbox{$\exists$open $P=(u_1,u_2,u_3,\dots)\in \+P_{(u_1),L}:$ type II $(u_2,u_3)$}\mid u_1\mbox{ open}}\notag\\
	&\le \frac{k^2d^2}{\alpha}\cdot 2^{-k}\cdot (1-\eps)^{\lfloor\frac{L-4}{2}\rfloor}.\label{eqn:typeII}
\end{align}

\subsection*{Putting all together} Plugging Equations (\ref{eqn:good}), (\ref{eqn:typeI}) and (\ref{eqn:typeII}) into \Cref{eqn:3-cases}, we have
\begin{align*}
	&\phantom{{}={}}\Pr{\mbox{$\exists$open $P\in \+P_{(u_1),L}$}\mid u_1\mbox{ open}}\\
	&\le \tp{k^2d \cdot 2^{-(1-\alpha)k}+2k^2d \cdot 2^{-k} +\frac{k^2d^2}{\alpha}\cdot 2^{-k}}\cdot (1-\eps)^{\lfloor\frac{L-4}{2}\rfloor}\\
	&< \left(2^{\alpha k + 2} + \frac{d}{\alpha} \right) k^2 d \cdot 2^{-k} \cdot (1-\eps)^{\lfloor\frac{L-4}{2}\rfloor}
\end{align*}
where we used a very crude inequality $2^{\alpha k} + 2 < 2^{\alpha k + 2}$. Let us take $\alpha = \alpha(k) := \frac{1}{2} + \frac{2 - \log k}{k} \in [1/4,1]$ for any $k \geq 2$. Note that $2^{\alpha k+2} = 2^{4+k/2} / k$ by definition. Hence the above bound continues as:
\begin{align*}
	\cdots
	&\leq \left(2^{4+k/2} kd + 4 (kd)^2 \right) 2^{-k} \cdot (1-\eps)^{\lfloor\frac{L-4}{2}\rfloor} \\
	&\leq (1-\eps) \cdot (1-\eps)^{\lfloor\frac{L-4}{2}\rfloor} \\
	&\leq (1-\eps)^{\lfloor \frac{L-2}{2} \rfloor}
\end{align*}
where the second line is due to our assumption $d \leq  \frac{(1-\eps)}{20}\cdot \frac{2^{k/2}}{k}$. 

% \begin{remark}
%     When $\eps := 0.1$, say, the condition becomes (approximately) $d \leq 0.027 \cdot 2^{k/2}$. Taking smaller $\eps$ will allow larger applicable range of $d$, but it comes at the price of slowing down the sampler (by a multiplicative log-factor; see \Cref{thm:cftp}).
    
%     % If $k$ is sufficiently large we can further improve our bound to $d \lesssim \frac{\sqrt{33-\eps}-1}{16} \cdot 2^{k/2}$.
% \end{remark}

\subsection{Improved bound for linear hypergraphs}

In this section, we show that our perfect sampler for linear hypergraphs is efficient when $d \leq \frac{(1-\eps)}{4} \cdot \frac{2^k}{k^2}$. \Cref{thm:sym-main-liner} is a consequence of the following lemma.

 \begin{lemma}\label{lem:percolation-sym-linear}
 	For all $\eps\in(0,1)$ and $k \geq 2$, if $\Phi$ is a linear hypergraph and
	\[
		d \leq \frac{(1-\eps)}{4} \cdot \frac{2^k}{k^2}
	\]
	then $\Pr{\B_L}\le \frac{k}{2^k}\cdot m \cdot (1-\eps)^{\lfloor\frac{L}{2}-1\rfloor}$ for all $L \in \^N$.
 \end{lemma}

     The proof is similar to the proof of \Cref{lem:percolation-sym}. We only need to show that
     \begin{lemma}\label{lem:symlinear}
        For all $\eps\in(0,1)$ and $k \geq 2$, if $\Phi$ is a linear hypergraph and 
    	\[
    		d \leq \frac{(1-\eps)}{4} \cdot \frac{2^k}{k^2}
    	\]
     then for every $u_1\in V_{H_T}$ and every $L \in \^N$,
     \[
         \Pr{\mbox{$\exists$open $P\in \+P_{(u_1),L}$} \mid u_1\mbox{ open}}\le (1-\eps)^{\lfloor\frac{L}{2}-1\rfloor}
     \]
     \end{lemma}
    
    \begin{proof}
        Let us take $\alpha := 1/k$. Improvement in \Cref{lem:symlinear} over \Cref{lem:sym} comes from the absence of type II branchings: Two hyperedges share at most one vertex in linear hypergraphs, meaning that if $u_2$ is not good then $C(u_2)=C(u_1)$. Hence
        \begin{align}
        	&\phantom{{}={}}\Pr{\mbox{$\exists$open $P\in \+P_{(u_1),L}$}\mid u_1\mbox{ open}}\label{eqn:2-cases}\\
        	&\le \Pr{\mbox{$\exists$open $P=(u_1,u_2,u_3,\dots)\in \+P_{(u_1),L}:$ good $u_2$}\mid u_1\mbox{ open}}\notag\\
        	&\quad+\Pr{\mbox{$\exists$open $P=(u_1,u_2,u_3,\dots)\in \+P_{(u_1),L}:$ type I $(u_2,u_3)$}\mid u_1\mbox{ open}}.\notag
        \end{align}
        Recall the calculations (\Cref{eqn:good} and \Cref{eqn:typeI}) in \Cref{refined}. We plug in our choice $\alpha = 1/k$ and obtain
        \begin{align*}
        	\Pr{\mbox{$\exists$open $P\in \+P_{(u_1),L}$}\mid u_1\mbox{ open}}
        	&\le \tp{k^2 d\cdot 2^{-(k-1)} + 2 k^2 d \cdot 2^{-k}}\cdot (1-\eps)^{\lfloor\frac{L-4}{2}\rfloor}\\
        	&\leq (1-\eps)\cdot (1-\eps)^{\lfloor \frac{L-4}{2} \rfloor}\\
                &=(1-\eps)^{\lfloor \frac{L}{2}-1 \rfloor}
        \end{align*}
        where the second line is due to our assumption $d \leq \frac{(1-\eps)}{4} \cdot \frac{2^k}{k^2}$.
    \end{proof}

\section{Acknowledgements}

Chihao Zhang was supported by National Natural Science Foundation of China Grant 61902241. 

%%
%% Bibliography
%%

%% Please use bibtex, 

\bibliography{refs}
\bibliographystyle{alpha}

\end{document}